\documentclass[a4paper, 12pt]{article}
\RequirePackage[authoryear]{natbib}
\RequirePackage{hyperref}

\usepackage{amsthm, amsmath, amssymb, mathrsfs, multirow, url, subfigure}
\usepackage{graphicx} 
\usepackage{ifthen} 
\usepackage{amsfonts}
\usepackage[usenames]{color}
\usepackage{fullpage}
\usepackage[normalem]{ulem}

\usepackage{color}

\newcommand{\model}{\mathcal{P}}

\newcommand{\RR}{\mathbb{R}}

\newcommand{\eps}{\varepsilon}

\newcommand{\lP}{\underline{P}}
\newcommand{\uP}{\overline{P}}

\newcommand{\uPi}{\overline{\Pi}}
\newcommand{\yobs}{y_{\text{obs}}}

\newcommand{\Mhat}{\widehat{M}}

\theoremstyle{plain} 
\newtheorem{thm}{Theorem}

\title{
  Stochastic Optimization for Numerical Evaluation of Imprecise Probabilities
}

\author{Nicholas Syring\footnote{Department of Statistics, Iowa State University, {\tt nsyring@iastate.edu}. } \; and \; Ryan Martin\footnote{Department of Statistics, North Carolina State University; {\tt rgmarti3@ncsu.edu}}}
\date{\today}

\begin{document}
\maketitle

\begin{abstract}
In applications of imprecise probability, analysts must compute lower (or upper) expectations, defined as the infimum of an expectation over a set of parameter values.  Monte Carlo methods consistently approximate expectations at fixed parameter values, but can be costly to implement in grid search to locate minima over large subsets of the parameter space.  We investigate the use of stochastic iterative root-finding methods for efficiently computing lower expectations.  In two examples we illustrate the use of various stochastic approximation methods, and demonstrate their superior performance in comparison to grid search.
  \smallskip

\end{abstract}

\section{Introduction}
\label{sec:intro}

Applications of imprecise probability require the evaluation of lower and upper expectations.  That is, let $\model$ denote a collection of probability distributions $P$ and, if $g$ is a $P$-integrable function for each $P \in \model$, then let 
\[ \lP g = \inf_{P \in \model} P g \quad \text{and} \quad \uP g = \sup_{P \in \model} P g \]
denote the lower and upper expectation of $g$, respectively, where $Pg = \int g(y) \, P(dy)$ is the ordinary expectation.  In practical problems, the collection $\model$ will typically be indexed by a finite-dimensional parameter $\theta \in \Theta \subseteq \RR^q$, $q \geq 1$, which provides both some simplicity and extra flexibility.  That is, let 
\[ M(\theta) = \int g_\theta(y) \, P_\theta(dy), \]
now for a class of functions $\{g_\theta: \theta \in \Theta\}$, indexed by $\Theta$, and consider the optimization problem 
\[ \inf_\theta M(\theta) \quad \text{and/or} \quad \sup_\theta M(\theta). \]

\citet{robbins.1951} provided a classical example of this problem.  Consider a linear regression model $M(x_i) = x_i^\top\beta + \epsilon_i$ where $y_i = M(x_i)$ is a response, $x_i$ is a vector of covariates, and $\epsilon_i$, $i=1,\ldots,n$, are independent, mean-zero random variables.  Usually, the goal is to estimate the parameter $\beta$ given observations $(y_i, x_i)$, $i=1,\ldots,n$.  However, the real interest is often to find the optimal input $x$ to produce a desired response $y$.  For example, farmers would be interested in maximizing crop yield by varying the amounts of fertilizers and pesticides applied to their fields.  In such a case, it is not necessary to assume linearity or any particular form of the regression function.  To maximize $M(x)$, one could implement a fixed design, which is akin to a (stochastic) grid search.  Alternatively, \citet{robbins.1951} provide an iterative method to maximize $M(x)$, discussed in Section~2, which can be understood as a method of optimal sequential design in regression models.    

Another general example arises in the context of statistical inference, and the so-called inferential model framework presented in \citet{martin.liu.book} and \citet{martin.liu.poss}.  Those methods rely on nested random sets or, alternatively, possibility measures, which require maximizing the expectations of certain (data- and) $\theta$-dependent functions.  For example, one relatively simple inferential model construction \citep[see, e.g.,][]{martin.2015, martin.2018} defines a plausibility contour 
\[ \pi(\theta; \yobs) = P_\theta\{T(Y,\theta) > T(\yobs, \theta)\}, \]
where $\yobs$ denotes the observed data, $Y$ is a random element having distribution $P_\theta$, and $T(y,\theta)$ is a suitably chosen scalar-valued function.  See Section~\ref{ss:logistic} for an example.  If $A$ is some assertion or hypothesis about the unknown value of $\theta$, the inferential model's upper probability, given the observed data $\yobs$, is the supremum 
\[ \uPi(A; \yobs) = \sup_{\theta \in A} \pi(\theta; \yobs), \]
which is a special case of the general problem presented above, with $g_\theta(y) = 1\{T(y,\theta) > T(\yobs, \theta)\}$.

Recently, several authors explored Monte Carlo sampling methods for evaluating a lower or upper expectation.  The basic idea is Monte Carlo with a grid search.  That is, let $\theta_j$, $j \in 1,\ldots,J$ denote a finite subset of $\Theta$ and, for each $j$, let 
\[ \Mhat(\theta_j) = \frac{1}{N} \sum_{n=1}^N g_{\theta_j}(Y_n^{(j)}), \quad Y_n^{(j)} \sim P_{\theta_j}, \]
be a Monte Carlo estimate of $M(\theta_j)$.  Then the lower and upper expectations can be readily approximated by $\min_j \Mhat(\theta_j)$ and $\max_j \Mhat(\theta_j)$, respectively. 

\citet{decadt.2019} show such Monte Carlo estimates of lower expectations are consistent.  In practice, however, the accuracy of these Monte Carlo approximations depends on the smoothness of the function and the fineness of the grid. As a consequence, it may be computationally prohibitive to produce enough Monte Carlo samples for a good approximation, especially when the parameter is multi-dimensional.  To lighten the computational burden, \citet{fetz.2016} and \citet{fetz.2019} propose reusing Monte Carlo samples by weighted resampling.  

In contrast, we propose to replace grid search by an iterative procedure.  The familiar deterministic setting provides some useful background.  Suppose the function $M(\theta)$ is known, sufficiently smooth, and convex.  Then, Newton's method with updates
\begin{align*}\theta_{t+1} = \theta_t - \{\ddot M(\theta_t)\}^{-1} \dot M(\theta_t)\end{align*}
converges quadratically to the minimizer $\theta^\star$, where the dots denote differentiation with respect to $\theta$.  Alternatively, the gradient descent update
\[\theta_{t+1} = \theta_t - \eps \dot M(\theta_t)\]
does not require the second derivative, and converges for a small enough step size $\eps>0$ and an initial point $\theta_0$ close enough to $\theta^\star$.  An important advantage of this approach is that it is less sensitive to the dimension of the optimization problem than a grid search.  However, in our present context, the function we seek to optimize, and hence its gradient, are not available in closed-form, so Newton's method cannot be applied directly.  Is there a stochastic version?  

\citet{robbins.1951} proposed a stochastic analog of Newton's method and showed that, under certain conditions, its sequence of iterates converges in probability to $\theta^\star$.  Roughly speaking, the Robbins--Monro procedure is a gradient descent algorithm but with the known derivative replaced by a (crude) Monte Carlo approximation thereof.  In this paper, we consider the use of the Robbins--Monro algorithm and its variants for numerical evaluation of lower and upper expectations.  

In Section~\ref{S:a} we discuss existing results that characterize problems suitable for stochastic optimization.  Section~\ref{s:b} reviews a number of variations on the classical stochastic optimization algorithm and discusses practical considerations such as the choice of step size $(\eps_t)$.  Section~\ref{s:examples} presents two relevant examples comparing the performance of various stochastic optimization methods with that of grid search.  Some concluding remarks are given in Section~\ref{S:discuss}.

\section{Algorithm and its Properties}
\label{S:a}

\subsection{Robbins--Monro Algorithm}

For the analysis in this section, suppose that the optimization problem is sufficiently regular that it can be recast as a root-finding problem.  First, let $P_\theta$ have a density $p_\theta$ with respect to, say, Lebesgue measure, and that $\theta \mapsto p_\theta(x)$ is differentiable for all $x$ except perhaps in a set of Lebesgue measure zero.  Second, suppose that $M$ is differentiable and that the order of differentiation and expectation can be interchanged, i.e., 
\[ R(\theta) := \dot M(\theta) = \int h_\theta(x) \, p_\theta(x) \, dx, \]
where $h_\theta(x) =  \dot g_\theta(x) + g_\theta(x) \dot\ell_\theta(x)$ and  $\ell_\theta(x) = \log p_\theta(x)$.  Then minimizing $M$ corresponds to finding a root of $R$.  Keep in mind that $R$ is a vector-valued function.

For root-finding, the basic Robbins--Monro algorithm proceeds as follows.  For an arbitrary $\theta_0$, define the updates
\begin{equation}
    \label{eq:robbins}
    \theta_{t+1} = \theta_t - \eps_{t+1} X_{t+1}, \quad t \geq 0, 
\end{equation}
where $(\eps_t)$ is a vanishing, deterministic step size sequence and $(X_t: t \geq 1)$ is a sequence of random vectors such that 
\[ E(X_{t+1} \mid X_1,\ldots,X_t) = R(\theta_t). \]

For some intuition as to why the Robbins--Monro algorithm works, consider the following heuristics.  For simplicity, let $q=1$.  First, try to minimize $M(\theta)$ by iterating gradient descent: $\theta_{t+1} = \theta_t - Y_t$ where $Y_t = \dot M(\theta_t) + \delta_t$ represents a noisy estimate of $\dot M(\theta_t)$.  Then, by substitution, the gradient descent update satisfies $\theta_{t+1} = \theta_t - \dot M(\theta_t) - \delta_t$.  In general, the noise term $\delta_t$ has approximately constant variance, so the iterates $\theta_t$ cannot converge to $\theta^\star$.  On the other hand, if we apply the Robbins-Monro update the noise term $\delta_t$ is multiplied by the vanishing step-size and can safely be ignored for large enough $t$.  Taking the argument one step further, suppose $\dot M(\theta_t) <0$ for $\theta_t < \theta^\star$ and $\dot M(\theta_t) >0$ for $\theta_t > \theta^\star$.  For $\theta_t<\theta^\star$ the iterates satisfy
\begin{align*}
    E(\theta_{t+1} \mid \theta_t) &= \theta_t - \epsilon_t\dot M(\theta_t) - \epsilon_tE(\delta_t)\\
    & = \theta_t - \epsilon_t\dot M(\theta_t)\\
    & > \theta_t;
\end{align*}
the iterates satisfy the reverse inequality when $\theta_t > \theta^\star$. Therefore, the {\em expected direction} of the next iteration is towards the minimizer $\theta^\star$. 

\subsection{Convergence Properties}

The claim is that, under certain conditions, the sequence $(\theta_t)$ defined by the Robbins--Monro algorithm converges almost surely to the root $\theta^\star$ of $R$ as $t \to \infty$. A particularly elegant proof of this convergence property is based on the following almost supermartingale convergence theorem of \citet{robbins.1971}.  Given a probability space and an increasing sequence $\{\mathcal{A}_t: t \geq 1\}$ of sub-$\sigma$-algebras on that space, they define a sequence of random variables $\{Z_t: t \geq 1\}$ to be an {\em almost supermartingale} if there exists non-negative random variables $\beta_t$, $\xi_t$, and $\zeta_t$ such that 
\begin{equation}
\label{eq:asm}
E(Z_{t+1} \mid \mathcal{A}_t) \leq (1 + \beta_t) Z_t + \xi_t - \zeta_t, \quad t \geq 1. 
\end{equation}
Then their Theorem~1 states that, if both $\sum_t \beta_t$ and $\sum_t \xi_t$ converge almost surely, then 
\[ \text{$\lim_{t \to \infty} Z_t$ exists} \quad \text{and} \quad \sum_{t=1}^\infty \zeta_t < \infty \quad \text{almost surely}. \]

To relate this to the Robbins--Monro algorithm, write $Z_t = \|\theta_t - \theta^\star\|^2$, where $\|\cdot\|$ denotes the usual $\ell_2$-norm on $\RR^q$.  Then we set $\mathcal{A}_t$ to be the $\sigma$-algebra generated by $\{Z_1,\ldots,Z_t\}$, take conditional expectation, and simplify.  The following theorem is a slight generalization of Application~2 in \citet{robbins.1971}.

\begin{thm}
\label{thm:rm}
Suppose the function $R$ and its root $\theta^\star$ satisfy 
\begin{equation}
\label{eq:root}
(\theta - \theta^\star)^\top R(\theta) \geq 0, \quad \text{for all $\theta \in \RR^q$}, 
\end{equation}
with equality if and only if $\theta=\theta^\star$.  Moreover, assume that there exists positive $a$ and $b$ such that 
\begin{equation}
\label{eq:Xbound}
E(\|X_{t+1}\|^2 \mid \mathcal{A}_t) \leq a + b\|\theta_t\|^2. 
\end{equation}
If the positive step size sequence $\{\eps_t\}$ satisfies 
\begin{equation}
\label{eq:decay}
\sum_{t=1}^\infty \eps_t = \infty \quad \text{and} \quad \sum_{t=1}^\infty \eps_t^2 < \infty, 
\end{equation}
then the Robbins--Monro sequence in  \eqref{eq:robbins} satisfies $\theta_t \to \theta^\star$ almost surely as $t \to \infty$.  
\end{thm}

\begin{proof}
First, it is easy to see that 
\begin{align*}
Z_{t+1} & = \|\theta_t - \eps_{t+1} X_{t+1} - \theta^\star\|^2 \\
& = Z_t + \eps_{t+1}^2 \|X_{t+1}\|^2 - 2 \eps_{t+1} (\theta_t - \theta^\star)^\top X_{t+1}.
\end{align*}
Taking conditional expectation gives 
\begin{align*}
E(Z_{t+1} \mid \mathcal{A}_t) & = Z_t + \eps_{t+1}^2 E(\|X_{t+1}\|^2 \mid \mathcal{A}_t) \\
& \qquad -2\eps_{t+1} (\theta_t - \theta^\star)^\top R(\theta_t).
\end{align*}
By \eqref{eq:Xbound} and the inequality $(a + b)^2 \leq 2(a^2 + b^2)$, 
\[ E(\|X_{t+1}\|^2 \mid \mathcal{A}_t) \leq a + 2bZ_t + 2b\|\theta^\star\|^2. \]
Therefore, if we set 
\begin{align*}
\beta_t & = 2b\eps_{t+1}^2 \\
\xi_t & = 2b\|\theta^\star\|^2 \eps_{t+1}^2 \\
\zeta_t & = 2 \eps_{t+1} (\theta_t - \theta^\star)^\top R(\theta_t),
\end{align*}
then they are all positive and inequality \eqref{eq:asm} holds.  By the second condition in \eqref{eq:decay}, it follows that both $\beta_t$ and $\xi_t$ are summable, so it follows from Theorem~1 in \citet{robbins.1971} that $Z_t$ has a limit and $\sum_t \zeta_t$ is finite almost surely.  To see that the $Z_t$ limit must be 0, by the first condition in \eqref{eq:decay}, the only way $\sum_t \zeta_t$ could be finite is if $(\theta_t - \theta^\star)^\top R(\theta_t)$ vanishes almost surely, at least on a subsequence.  Since $Z_t$ has a limit, it is clear that the only way this latter claim holds is if the $Z_t$ limit is 0.
\end{proof}

Stronger and more modern results for almost sure convergence of the Robbins--Monro process are available.  For a glimpse of what modern proofs entail, consider rewriting update \eqref{eq:robbins} by adding and subtracting $\eps_{t+1} R(\theta_t)$ to get
\[ \theta_{t+1} = \theta_t - \eps_{t+1} R(\theta_t) - \eps_{t+1}\{ X_{t+1} - R(\theta_t)\}. \]
The latter term is a zero-mean martingale difference sequence and, under mild conditions, would converge to 0 almost surely.  Ignoring the negligible martingale term and rewriting the above relationship, we get 
\[ \eps_{t+1}^{-1} (\theta_{t+1} - \theta_t) \approx R(\theta_t). \]
The left-hand side resembles a derivative of ``$t \mapsto \theta_t$,'' so there is a close connection between the asymptotic properties of the Robbins--Monro process and that of solutions to the ordinary differential equation ``$\frac{d}{dt} \theta_t = R(\theta_t)$.''  See \citet{martin.2008} for an overview, and \citet{kushner.2003} for a comprehensive account.  

Besides almost sure convergence there are results characterizing the random behavior of $(\theta_t - \theta^\star)$ for large $t$.  An early reference is \citet{sacks.1958}, which shows that under certain regularity conditions $t^{1/2}(\theta_t - \theta^\star)$ converges in distribution to a normal random variable with variance proportional to $\epsilon_0^2$ where the step size is given by $\epsilon_t = \epsilon_0 t^{-1}$.

\section{Variations on Robbins--Monro}
\label{s:b}

\subsection{Finite Differences}
\label{s:b1}

A key ingredient of the update in \eqref{eq:robbins} is the identification of an unbiased estimator of the gradient $\dot M(\theta)$.  In general, it may be hard to find such an estimator.  Instead, the update can be modified to use a finite-difference approximation.  Let $e_i$ denote the unit vector in direction $i$ and define 
\[ Y_{t, i} = \frac{g(X^+_i)-g(X^-_i)}{2c_t}, \]
for $c_t>0$ fixed or vanishing and where 
\[ X_i^+\sim P_{\theta_t + c_t e_i} \quad \text{and} \quad X_i^-\sim P_{\theta_t - c_t e_i}. \]
Setting $Y_{t} = (Y_{t,1}, \ldots, Y_{t,q})$ the modified update is
\begin{equation}
    \label{eq:kiefer}
    \theta_{t+1} = \theta_t - \eps_tY_t.
\end{equation}
\citet{kiefer.1952} originated the update in \eqref{eq:kiefer} and proved consistency of ${\theta_t}$.  

To see how the finite-difference approximation may affect the behavior of the algorithm we'll rewrite the update in \eqref{eq:kiefer}.  Define the random Monte Carlo errors 
\[ \psi_{t,i} = M(\theta_t+c_te_i) - g(X_i^+) - \{M(\theta_t-c_te_i) - g(X_i^-)\}, \]
which are simply differences between the true function values and noisy estimates obtained from Monte Carlo sampling.  Let 
\[ \dot M_i(\theta_t) - \beta_{t,i} = \frac{M(\theta_t+c_te_i)-M(\theta_t-c_te_i)}{2c_t}, \]
so that $\beta_{t,i}$ denotes the error from approximating the derivative by the finite difference.  Then, the update in \eqref{eq:kiefer} can be written
\[\theta_{t+1} = \theta_t + \eps_t\dot M(\theta_t) - \eps_t(\psi_t/2c_t) - \eps_t\beta_t\]
where $\psi_t = (\psi_{t,1},\ldots, \psi_{t,q})$ and $\beta_t = (\beta_{t,1},\ldots, \beta_{t,q})$.  As with the Robbins--Monro update, the asymptotic behavior of the Kiefer--Wolfowitz algorithm follows the behavior of the differential equation $d\theta/dt = \dot M(\theta_t)$ provided the error terms $\eps_t \psi_t/(2c_t)$ and $\eps_t\beta_t$ vanish.  According to \citet{kushner.2003}, $\beta_t = O(c_t)$, so the choice of $c_t$ presents a kind of bias-variance trade-off.  A large $c_t$ results in worse approximation of the derivative $\dot M(\theta_t)$, but lowers the Monte Carlo noise term $\psi_t/(2c_t)$.  

The variance of the Monte Carlo noise term can be minimized by maximizing the covariance of $g(X_i^+)$ and $g(X_i^-)$. One method that may improve the practical performance of the Kiefer--Wolfowitz update is to use correlated Monte-Carlo samples.  For example, if $P_\theta$ can be sampled by an inverse-CDF transform, then we could sample $U\sim \text{Unif}(0,1)$ and let $X_i^{\pm} = P_{\theta\pm c_t e_i}^{-1}(U)$.    

\citet{sacks.1958} also showed asymptotic normality of the Kiefer--Wolfowitz iterates under regularity conditions.  If the step size is taken to be $\epsilon_t = \epsilon_0 t^{-1}$, then $t^{1/2} c_t(\theta_t - \theta^\star)$ converges in distribution to a normal random variable with variance proportional to $\epsilon_0^2$.  \citet{sacks.1958} specifies $c_t\rightarrow 0$, which implies the Kiefer--Wolfowitz procedure converges more slowly than the Robbins--Monro procedure; in other words, there is a significant cost to approximating derivatives by finite differences.      

\subsection{Parameter Constraints}
\label{s:b2}

When estimating lower expectations it is natural to confine the parameter space to compact subsets; and see the example in Section~\ref{ss:ex1}.  Updates \eqref{eq:robbins} and \eqref{eq:kiefer} offer no guarantee the iterates $\theta_t$ will remain within any finite neighborhood of the initial point $\theta_0$.  A simple fix projects the algorithm to the constraint space.  Let $\Theta_C\subset \Theta$ denote a compact constraint space, usually a rectangle in $\mathbb{R}^q$, and let $d:\mathbb{R}^q \times \mathbb{R}^q\mapsto \mathbb{R}^+$ denote a metric on $\mathbb{R}^q$.  Let $\mathsf{proj}_d(\theta)$ equal the value of $\theta'$ minimizing $d(\theta,\theta')$ over $\theta'\in \Theta_C$.  Then, the projected update has the form
\[\theta_{t+1} = \mathsf{proj}_d\{\theta_t - \eps_t Y_t\}.\]

\subsection{Averaging Iterates}
\label{s:b3}

In practice the analyst applies update \eqref{eq:robbins} or \eqref{eq:kiefer} until they meet some pre-specified convergence criteria, at which point they report the most up-to-date iterate $\theta_t$.  It turns out that it may be advantageous instead to report the average of the iterates $\bar\theta = t^{-1}\sum_{s=1}^t \theta_s$ as the final estimate.  In practice, the updates generally converge more quickly with larger step sizes $\eps_t$.  But, large step sizes also increase iterate variability, and averaging the iterates naturally reduces this variability.  Furthermore, the practical benefit of iterate averaging holds up in theory.  When $\eps_t = O(t^{-1})$ it can be shown the mean squared error $E\|\theta_t-\theta^\star\|_2^2$ behaves like $t^{-1}$, but the corresponding mean squared error for the averaged iterates $\bar \theta_t$ vanishes like $t^{-1}$ even when $\eps_t=O(t^{-1/2})$; see \citet{polyak} and \citet{ kushner.2003}.

\subsection{Averaging or Bounding Monte Carlo Samples}
\label{s:b4}

In addition to averaging iterates, it may be helpful to average over $M\geq 1$ Monte Carlo samples at each iteration.  One reason to take $M>1$ samples is when the function $M(\theta)$ likely may be zero, for example, when $M(\theta)$ is the probability of a rare event.  In that case, the estimate of the gradient could vanish, which may trigger the convergence criteria or simply cause the iterates to get stuck at a constant value.  For the Robbins--Monro update \eqref{eq:robbins}, replace $Y_t$ by $\overline Y_t = M^{-1}\sum_{j=1}^M Y_{t}^j$ where each $Y_{t}^j$ is an independent sample with mean $\dot M(\theta_t)$.  For the Kiefer--Wolfowitz update \eqref{eq:kiefer}, the finite difference approximation can be computed by the Monte Carlo average $\overline Y_t = M^{-1}\sum_{j=1}^M Y_{t}^j$, where $Y_{t}^j = (Y_{t,1}^j,\ldots, Y_{t,q}^j)$, 
\[ Y_{t,i}^j = \frac{g(X^{+,j}_i)-g(X^{-,j}_i)}{2c_t}, \]
and $X^{\pm,j}_i\sim P_{\theta_t\pm c_t e_i}$, for $j=1,\ldots,M$.  The obvious drawback to observation averaging is that it requires many more Monte Carlo samples.

In other cases, the function $M(\theta)$ may be highly sensitive to $\theta$ so that there is a chance to generate an extreme update.  \citet{kushner.2003} suggest upper-bounding the absolute value of Monte Carlo samples, or, equivalently, upper-bounding by a constant the absolute change in subsequent iterate values.  

\subsection{Second-Order Methods}

The convergence rate of iterative methods for optimizing deterministic functions usually improves when those methods employ second-order information, like a Hessian matrix.  Perhaps surprisingly, stochastic second-order methods can achieve faster rates than their first-order counterparts only by a constant multiple \citep{agarwal.etal}.  The stochastic or Robbins--Monro analog of the classical Newton method updates via
\begin{equation}
\label{eq:2nd}
    \theta_{t+1} = \theta_t - \eps_tZ_tY_t,
\end{equation}
where $Z_t$ is an unbiased estimator of $[\ddot M(\theta_t)]^{-1}$.  As with the first-order Robbins--Monro algorithm, it may be challenging to find an unbiased estimator of the inverse Hessian.  Similar to the Kiefer--Wolfowitz approach, the Hessian can be estimated by finite differences, but there is a more efficient approach.  In deterministic function optimization the BFGS methods \citep[e.g.,][Chap.~3.4]{fletcher1987} are computationally efficient alternatives to computing matrix inverses.  These methods iteratively update $Z_t$ by
\begin{align*}
\delta_t & = \theta_{t+1}-\theta_t \\
\gamma_t & = Y_{t+1}-Y_t\\
Z_{t+1} & = \left(I - \frac{\gamma_t \delta_t^\top}{\delta_t^\top \gamma_t}\right)^\top Z_t \left(I - \frac{\gamma_t \delta_t^\top}{\delta_t^\top \gamma_t}\right) + \frac{\delta_t \delta_t^\top}{\delta_t^\top \gamma_t}.
\end{align*}
\citet{byrd.etal} developed a stochastic version of the BFGS algorithm and demonstrated its performance in machine learning problems.

\subsection{ADADELTA and Choice of Step Size}
\label{s:b6}

All three updates discussed above rely on a user-specified step size or learning rate $\eps_t$.  When an average of iterates will be reported often $\eps_t = \eps_0 t^{-\tau}$ where $\tau \in (1/2,3/4)$.  In practice, the leading constant $\eps_0$ can have a surprisingly strong impact on the speed of convergence of the iterates.  When $\eps_0$ is too small the sequence $\{\theta_t\}$ may change very little and practically fail to converge because the user's maximum number of iterations is exceeded.  On the other hand, when $\eps_0$ is very large there is excessive variation early in the sequence $\{\theta_t\}$, and this may cause the sequence of iterate averages to converge slowly.

Finding the optimal value $\eps_0$ is a challenge.  One strategy is to run a small number of iterations, say $10$, for several values of $\eps_0$ and choose the smallest value of $\eps_0$ such that a measure of variation (e.g., range or variance) of the corresponding $10$ iterates is sufficiently large.     

Alternative updating formulas, such as the ADADELTA method in \citet{zeiler}, do not require a user-specified step size at all.  Instead, ADADELTA  iteratively and adaptively computes a step size from the change in successive iterates and a running average of the gradient.  Select an averaging parameter close to one, e.g., $\rho = 0.995$.  Define the root mean square function $RMS(x) = (x^2 + \eta)^{1/2}$ for some small stabilizing constant, e.g., $\eta=0.995$.  The ADADELTA algorithm makes an initial update based only on the gradient, as in \eqref{eq:robbins}, and for $t>1$ updates according to:
\begin{enumerate}
    \item Compute gradient estimate $Y_t$
    \item Accumulate gradient $S_t = \rho S_{t-1} + (1-\rho)Y_t^2$
    \item Compute update $\delta \theta_t = -\frac{RMS(D_{t-1})}{RMS(S_t)}Y_t$
    \item Accumulate update $D_t = \rho D_{t-1} + (1-\rho)(\delta\theta_t)^2$
    \item Apply update $\theta_{t+1} = \theta_t + \delta \theta_t$
\end{enumerate}

\section{Examples}
\label{s:examples}

\subsection{A Gaussian Probability}
\label{ss:ex1}

The first example is taken from \citet{fetz.2019}.  Consider estimating the lower expectation $\min_\theta M(\theta)$ where $M(\theta) = E[1\{X\notin (-2,2)\}]$ and where $X\sim N(\theta, \sigma = 2)$.  The minimum of $M(\theta)$ occurs at $\theta = 0$.  Stochastic optimization techniques experience the greatest gains over grid search Monte Carlo estimation when the parameter space is multi-dimensional.  Nevertheless, this simple example illustrates the practical differences between the variations of stochastic optimization described in Section~\ref{s:b}.  

First, consider the Robbins--Monro update in \eqref{eq:robbins} run for $1000$ iterations.  Figure~\ref{fig:1} displays the iterates and average iterates for two learning rates: $t^{-1/2}$ and $5t^{-1/2}$.  There are two important features of the plot.  The larger learning rate produces much more variation in the iterates, which helps the sequence to quickly find the minimum and then randomly vary around that minimum.  The bias in the iterates coming from the initial point $\theta_0$ quickly dissipates and the average of iterates quickly settles down near $\theta = 0$.  In contrast, the sequence of iterates with learning rate $t^{-1/2}$ move very slowly towards $\theta = 0$; so slowly, in fact, that there is no benefit to averaging the iterates since the average remains biased towards the initial point $\theta_0 = 6$.      

\begin{figure}[t]
    \centering
        \includegraphics[width = 0.5 \linewidth]{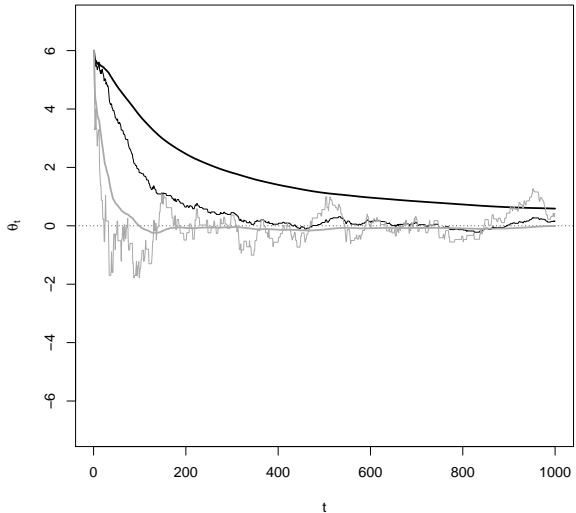}
    \caption{Sequences of iterates $\theta_t$ (fine lines) and their averages $\overline \theta_t$ (bold lines) for the Robbins--Monro update \eqref{eq:robbins}.  Black lines indicate learning rate $\eps_t = t^{-1/2}$ while gray lines correspond to $\eps_t = 5t^{-1/2}$.}
    \label{fig:1}
\end{figure}

Next, we consider whether averaging over $M>1$ Monte Carlo samples may improve the performance of the Robbins--Monro update.  For the basic method, consider running the Robbins--Monro update in \eqref{eq:robbins} and with $\eps_t = 20t^{-1/2}$ for $1000$ iterations, and in each step using only $M=1$ Monte Carlo sample $Y_t$ to estimate the gradient.  For comparison, run the Robbins--Monro update only $100$ but use $M=10$ Monte Carlo samples in each iteration, and estimate the gradient by $\overline Y_t = \tfrac{1}{10}\sum_{j=1}^{10}Y_t^j$ for Monte Carlo samples $Y_t^1, \ldots, Y_t^{10}$.  The idea is to compare the performance of the two variations of \eqref{eq:robbins} for the same number of Monte Carlo samples.  Each method was run $1000$ times, and $100$ of those corresponding paths of iterate averages are displayed in Figures~\ref{fig:2}--\ref{fig:3}.  The figures do not suggest averaging Monte Carlo samples at each iteration speeds up convergence.  The last iterate average $\bar \theta_{100}$ for the averaging method was about $0.34$ on average and with standard deviation $0.44$ over $1000$ repetitions. In contrast, the last iterate average $\bar \theta_{1000}$ for the method that drew only one Monte Carlo sample per iteration has average value $0.01$ with standard deviation $0.19$ over $1000$ repetitions.  For reference, a simple grid search performing $10$ Monte Carlo samples at each of $100$ grid points performed worse than both methods, with an average solution of $-0.44$ and a standard deviation of $0.82$ over $1000$ repetitions. The takeaway is that averaging multiple Monte Carlo samples at each iteration does not speed up convergence.       

\begin{figure}[t]
    \centering
        \includegraphics[width = 0.5 \linewidth]{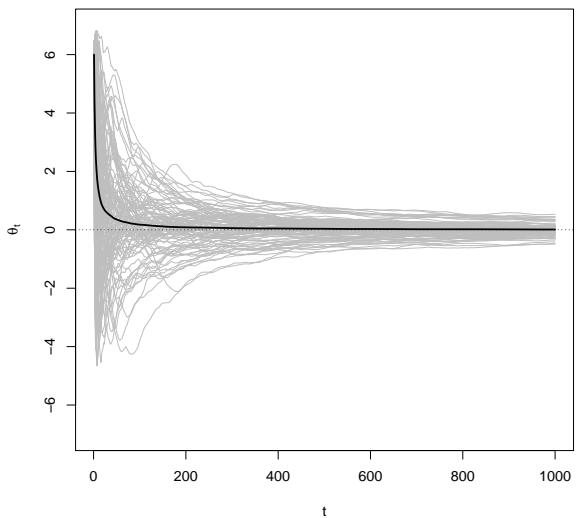}
        \caption{Averages of iterates $\bar \theta_t$ (gray lines) and their overall average (black line) over $1000$ repetitions of the Robbins--Monro update \eqref{eq:robbins} with no Monte Carlo averaging of samples.}
    \label{fig:2}
\end{figure}

\begin{figure}[t]
    \centering
        \includegraphics[width = 0.5 \linewidth]{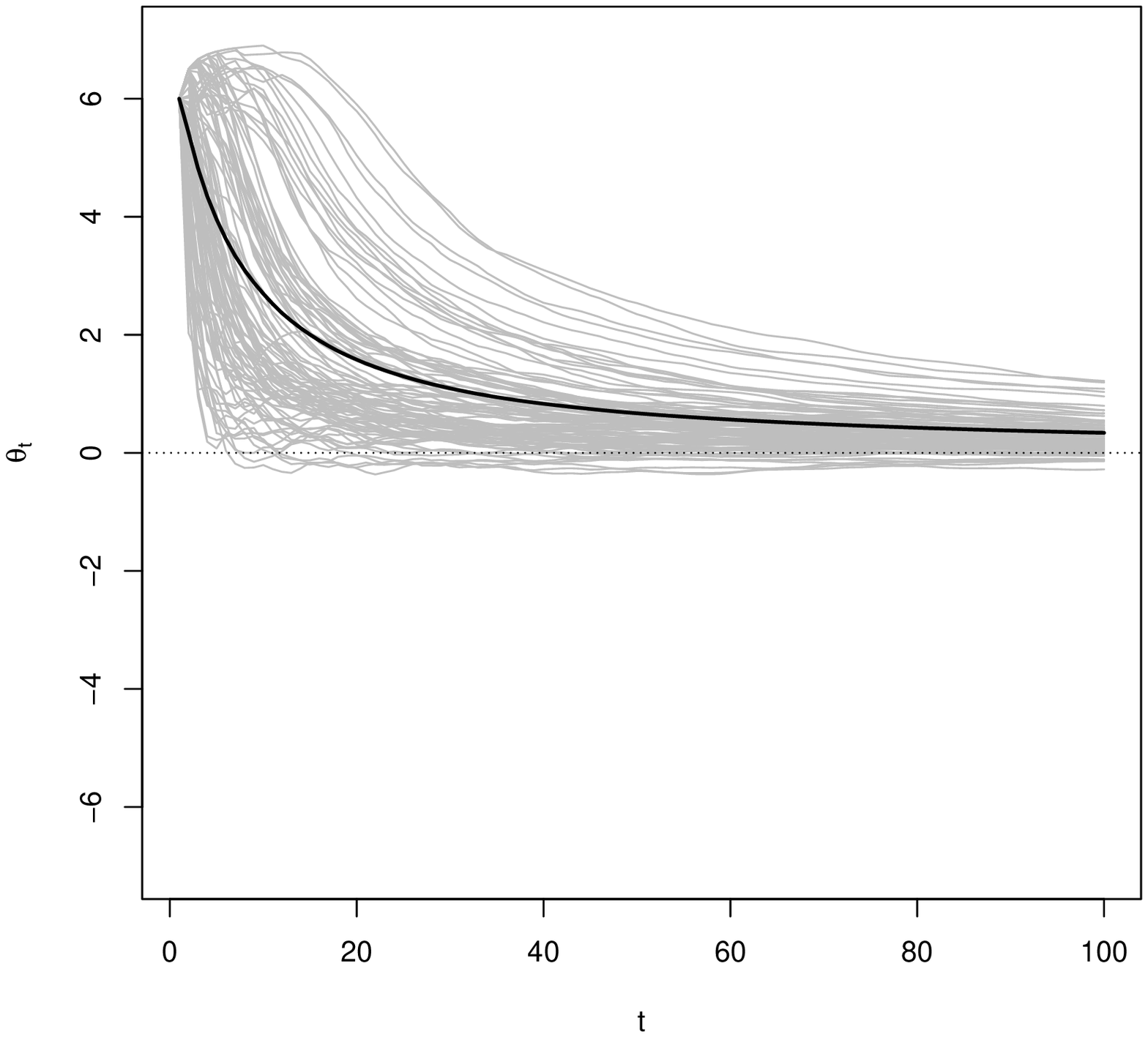}
        \caption{Averages of iterates $\bar \theta_t$ (gray lines) and their overall average (black line) over $1000$ repetitions of the Robbins--Monro update \eqref{eq:robbins} with averaging of $M=10$ Monte Carlo samples per iteration.}
    \label{fig:3}
\end{figure}

Besides Monte Carlo averaging, it is possible incorporating second derivative information into the Robbins--Monro update might improve its practical performance even if it provides no substantial benefit according to convergence theory.  However, in this example the Newton style of update in \eqref{eq:2nd} is very inefficient.  The trouble is that the update is highly sensitive to the second derivative, at least for the first several hundred iterations, and this sensitivity causes the algorithm to behave erratically.  One way to dampen the effect of high variation in the estimate of the second derivative is, of course, to use $M \gg 1$ Monte Carlo samples per iteration.  We used $M=50$ Monte Carlo samples to produce Figure~\ref{fig:4}, yet the iterate paths are still prone to excessive variation early in the sequence.  Figure~\ref{fig:5} shows the knock-on effect of this variation is iterate averages may become biased and slow to converge.  An alternative is to compute a moving average of the last $k$ iterates, rather than all the previous iterates.  Over $100$ repetitions the second-order updating method produced an average solution within $\pm0.01$ with a standard deviation of $0.08$, better than the first order methods tried above.  But, this is based on $50$ times the Monte Carlo samples, so does not represent an improvement over the Robbins--Monro method.      

\begin{figure}[t]
    \centering
        \includegraphics[width = 0.5 \linewidth]{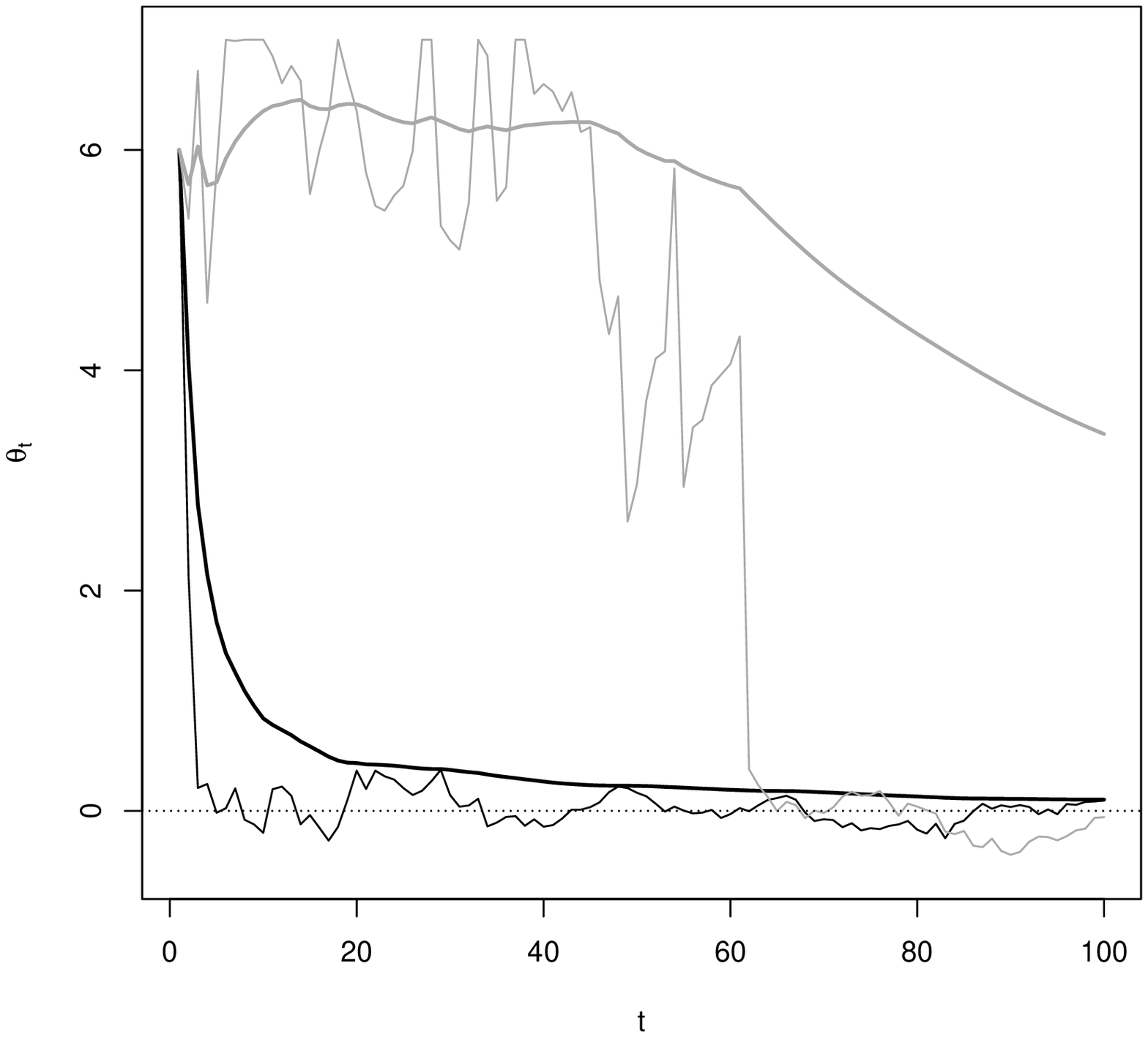}
        \caption{Two sequences of iterates $\theta_t$ (fine lines) and their averages $\bar \theta_t$ (bold lines) for the second-order update \eqref{eq:2nd}.}
    \label{fig:4}
\end{figure}

\begin{figure}[t]
    \centering
        \includegraphics[width = 0.5 \linewidth]{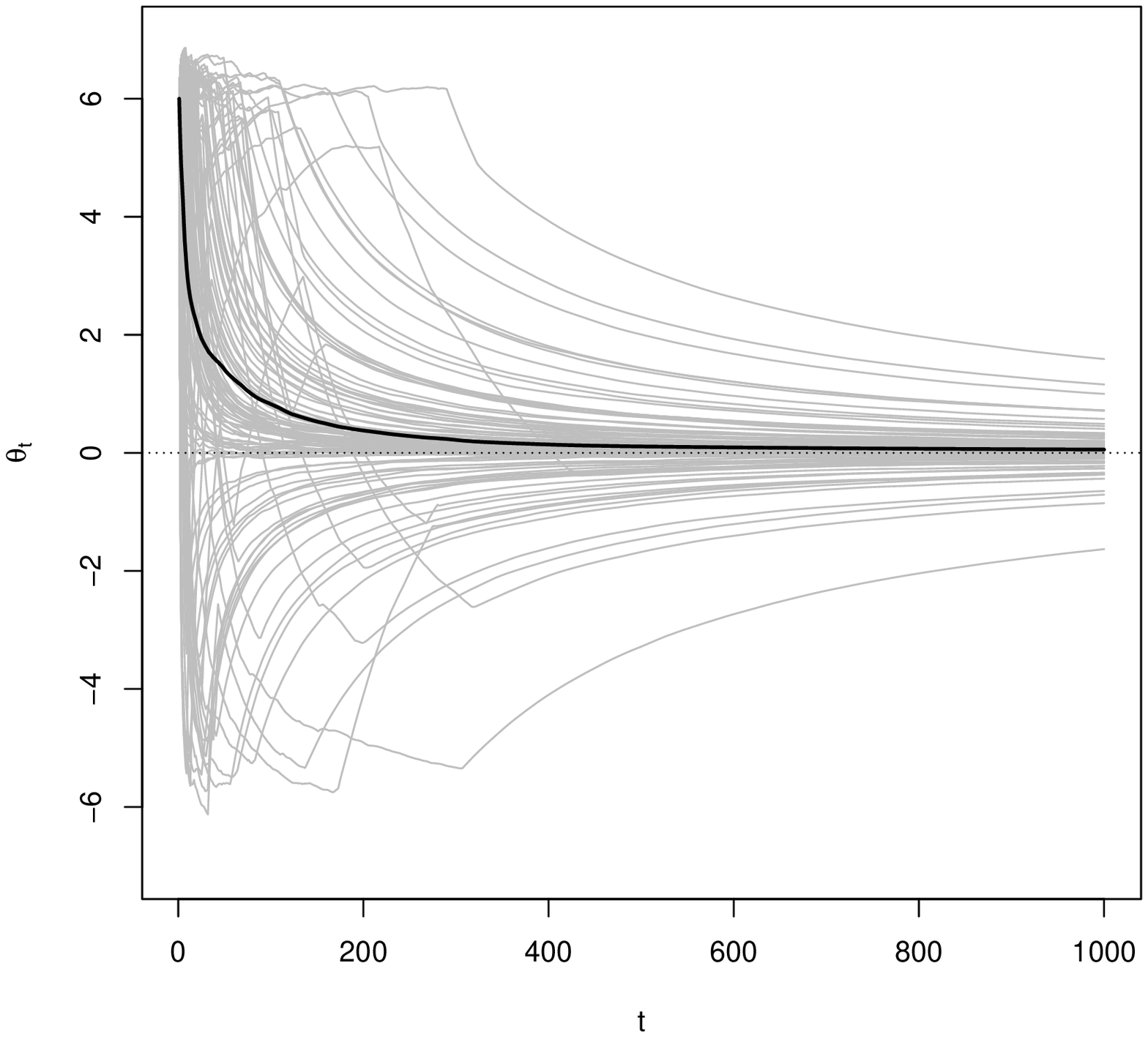}
        \caption{Averages of iterates $\bar \theta_t$ (gray lines) and their overall average (black line) over $1000$ repetitions of the second-order update \eqref{eq:2nd} with averaging of $M=50$ Monte Carlo samples per iteration.}
    \label{fig:5}
\end{figure}

So far, our investigation suggests the original Robbins--Monro update \eqref{eq:robbins} without Monte Carlo averaging of samples, but with averaging of iterates performs best provided a good choice of learning rate can be made.  Since it is not straightforward to choose a good learning rate in practice, we last consider the ADADELTA update of Section~\ref{s:b6} that makes an ``automatic" choice of learning rate.  Figure~\ref{fig:6} displays iterate and average iterate trajectories for the ADADELTA update.  After $10000$ iterations the average solution using the average of iterates is $0.19$ with a standard deviation of $0.50$; the median solution is $0.07$.  ADADELTA practically converges much more slowly than the Robbins--Monro update with a good choice of learning rate, but it may be an attractive option if the cost of using many iterations is not too high.      

\begin{figure}[t]
    \centering
    \includegraphics[width = 0.5 \linewidth]{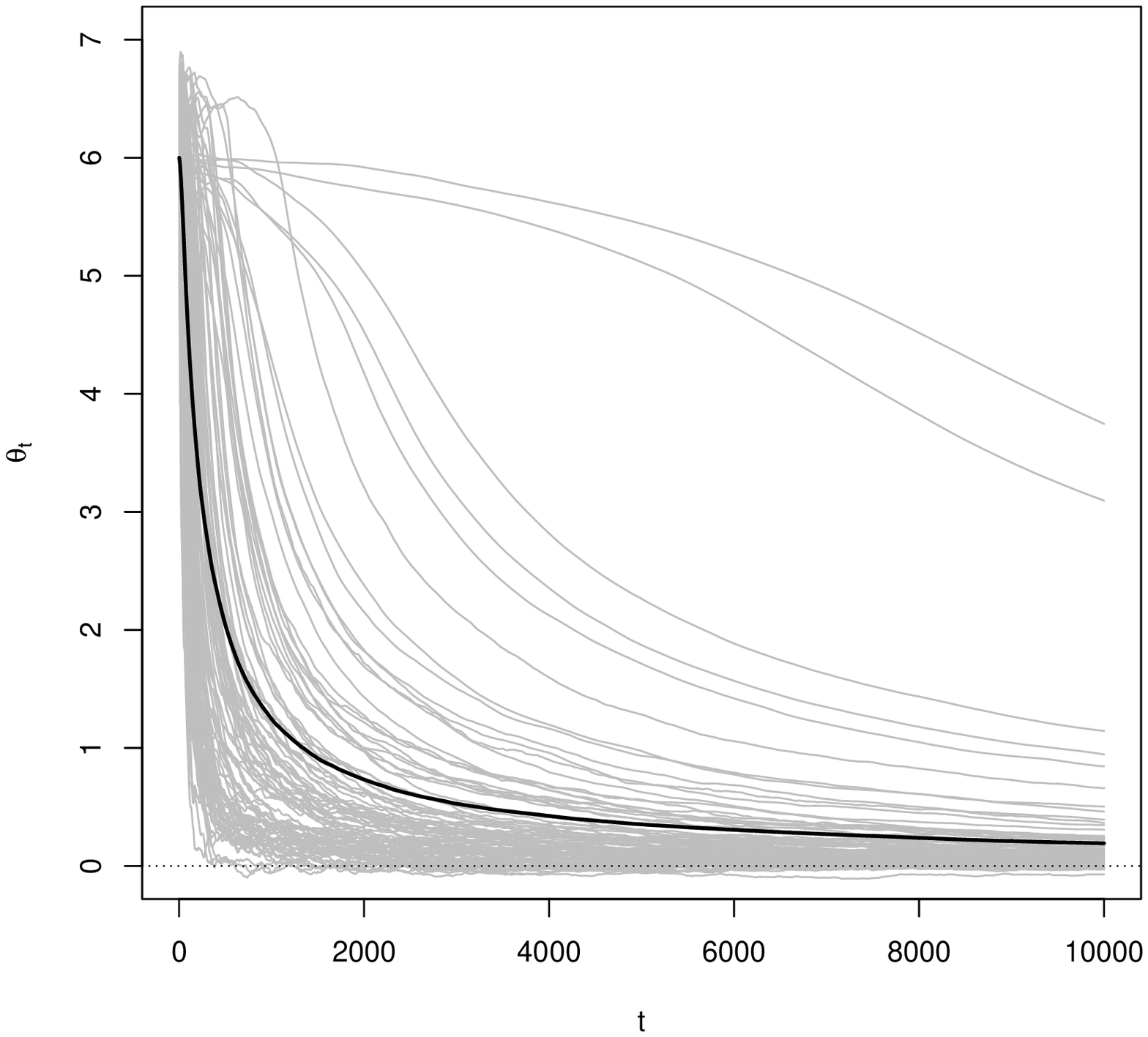}
    \caption{Averages of iterates $\bar \theta_t$ (gray lines) and their overall average (black line) over $1000$ repetitions of ADADELTA.}
    \label{fig:6}
\end{figure}

\subsection{Inferential Models in Logistic Regression}
\label{ss:logistic}

Suppose we have binary data $Y=(Y_1,\ldots,Y_n)^\top$, independent, where $Y_i$ is a Bernoulli random variable with parameter $p_i = F(x_i^\top \theta)$, where $x_i$ is a fixed/deterministic $q$-vector of explanatory variables, $\theta$ is a $q$-vector of coefficients, and $F$ is the logistic distribution function, $F(z) = (1 + e^{-z})^{-1}$.  One reason the logistic distribution is preferred for binary regression is its connection to the {\em odds} of $Y=1$, which satisfy
\[\log\left(\frac{p_i}{1-p_i}\right) = x_i^\top \theta.\]
Consequently, logistic regression parameters have a convenient interpretation similar to linear regression slope parameters: for a unit increase in $x_{ij}$ the logarithm of the odds increases by $\theta_j$.  Relevant questions like ``does predictor $x_j$ affect/increase/decrease the odds?'' can be investigated by evaluating suitable lower and upper probabilities to the respective assertions $A=\{\theta: \theta_j = 0\}$, $A=\{\theta: \theta_j > 0\}$, and $A=\{\theta: \theta_j < 0\}$.  

As discussed in Section~\ref{sec:intro}, if the goal is inference on $\theta$ based on observed data $Y=\yobs$, we can construct a generalized inferential model as follows.  Define 
\[ T(y,\theta) = -\log \{ L_y(\theta) / L_y(\hat\theta_y) \}, \]
where $L_y$ and $\hat\theta_y$ denote the likelihood function and the maximum likelihood estimator, respectively, based on a data set $y$, which depends implicitly on $F$ and the explanatory variables $x_1,\ldots,x_n$.  And, let $P_\theta$ denote the joint distribution of $Y$, which depends explicitly on the parameter $\theta$ and implicitly on the (deterministic) explanatory variables $x_1,\ldots,x_n$.  Then the plausibility contour is given by 
\[ \pi(\theta; \yobs) = P_\theta\{T(Y,\theta) > T(\yobs,\theta)\}, \]
and to evaluate the upper probability of an assertion $A$ about $\theta$ requires solving the optimization problem
\[ \uPi(A; \yobs) = \sup_{\theta \in A} \pi(\theta; \yobs). \]
One reason to compute this upper probability is to use it for evaluating a hypothesis test.  In the inferential modeling framework a level-$\alpha$ test of $H_0:\theta\in\Theta_0$ rejects if the inferential model upper probability $\uPi(\Theta_0; \yobs) < \alpha$, which is similar to a $p$-value-based rejection rule.  The general theory ensures, among other things, that the aforementioned test controls the frequentist Type~I error at level $\alpha$.  

Next, we describe two simulation experiments to evaluate the performance of stochastic optimization in computing inferential model upper probabilities for logistic regression.  Our simulations mimic experiments with fixed, complete designs.  In our first simulation we define $q=3$ predictors $x_j$, $j=1,2,3$, each taking values in $\{0,1/3,2/3,1\}$ and we form the $64 \times 3$ fixed design matrix with one row for every combination of predictor values.  Given predictor vector $x_i = (x_{i,j}, x_{i,2}, x_{i,3})^\top$, response $Y_i$ is sampled from a Bernoulli distribution with success probability $p_i=F(x_i^\top \theta^\star)$ for $\theta^\star = (-2,-1,2)^\top$ and where $F(z)$ denotes the logistic distribution function.  Let $A := \{\theta=(\theta_1, \theta_2, \theta_3)^\top: \theta_2>0\}$.  To quantify our uncertainty about the answer to the question ``does predictor $x_2$ increase the odds of success?,'' we compute $\uPi(A; \yobs)$.

For the second simulation we again consider a logistic regression model but now with $q=4$ predictors, each taking values in the set $\{0,1/2,1\}$, so that the complete, fixed design matrix has $3^4 = 81$ rows.  The true coefficients are $\theta^\star = (-2,-1,2,1)^\top$, and our assertion of interest is $A := \{\theta: \theta_2>0\}$.  

For the stochastic optimization approach we utilize the Kiefer--Wolfowitz update in \eqref{eq:kiefer}, which requires approximating the gradient of $\pi(\theta_t; \yobs)$ by finite differences.  A Monte Carlo approximation of $\pi(\theta_t; \yobs)$ is given by
\[\tilde\pi_N(\theta_t) = N^{-1}\sum_{i=1}^N 1\{T(Y_n^{(j)},\theta_t) > T(\yobs, \theta_t)\},\]
where $Y_n^{(j)} = (Y_{1,n}^{(j)},\ldots, Y_{n,n}^{(j)})^\top$ and 
\[Y_{i,n}^{(j)}\stackrel{\text{\tiny ind}}{\sim} \text{Ber}\{F(x_i^\top\theta_t)\}, \quad i=1,\ldots,n. \]
Then, the $i^\text{th}$ component of the gradient of $\pi(\theta_t; \yobs)$ is approximated by the following difference of Monte Carlo approximations,
\[\frac{\tilde\pi_N(\theta_t+c_te_i)-\tilde\pi_N(\theta_t-c_te_i)}{2c_t},\]
where $e_i$ denotes the unit vector in direction $i$.  We run two different variations of the Kiefer--Wolfowitz algorithm, both with a step-size of $30t^{-0.5}$ and a finite-difference radius of $c_t = t^{-0.5}$, but with different numbers of samples $N$ used in the finite difference approximation and different numbers of iterations $M$.  For the first run we set $N= 16$ and $M = 1250$, and for the second run we set $N = 40$ and $M = 500$.  We do not use any stopping rule for early termination of the updates; they run until reaching $M$ iterations.  For each Kiefer--Wolfowitz algorithm we compare the approximations of $\uPi(A; \yobs)$ using averaged iterates versus the final iterate.  Stochastic optimization theory implies we should average iterates since we use a large step size satisfying $\epsilon_t =O(t^{-0.5})$.  For comparison to stochastic optimization, we also approximate $\uPi(A; \yobs)$ by a grid search method, which we describe below.

We run each simulation on a total of $200$ randomly generated response vectors.  For simulation 1 the total number of Monte Carlo samples used is $N\times M\times n \times 6 = 7,680,000$, and for simulation 2 the total is $N\times M\times n \times 8 = 12,960,000$.  Note that $2\times N\times n$ Monte Carlo samples are needed to approximate each element of the gradient vector.  We use nearly the same number of Monte Carlo samples for grid search.  In the first simulation we fix a grid with $15^3=3375$ points on the set $\theta\in [-3,3]\times [0,3]\times [-3,3]$ and approximate $\uPi(A; \yobs)$ using $36$ Monte Carlo samples at each point for a total of $15^3\times 36\times n = 7,760,000$ samples.  For the second simulation we reduce the density on the grid to $10^4$ points in the set $\theta\in [-3,3]\times [0,3]\times [-3,3]\times [-3,3]$ and reduce the number of Monte Carlo samples used to approximate $\uPi(A; \yobs)$ at each point to $16$ for a total of $ 10^4 \times16\times n = 12,960,000$ Monte Carlo samples.  

In both simulations, we compare the various estimates of the upper probability $\uPi(A; \yobs)$ to the plausibility contour $\pi(\tilde\theta;\yobs)$ evaluated at $\tilde\theta$, the constrained maximum likelihood estimator (MLE) under the constraint $\theta_2 \geq 0$.  The reason for this is two-fold.  First, we know that $\tilde\theta$ maximizes $\theta \mapsto T(\yobs, \theta)$ over $A$ so there is good reason to think the maximizer of the $\theta$-dependent probability that defines $\pi(\theta;\yobs)$ would also be maximized, at least approximately, at $\tilde\theta$.  Second, even if $\pi(\tilde\theta; \yobs)$ is not an especially accurate approximation of $\uPi(A; \yobs)$, we do know that $\uPi(A; \yobs) \geq \pi(\tilde\theta; \yobs)$.  Therefore, between two estimates of $\uPi(A; \yobs)$, say, $\text{Est}_1$ and $\text{Est}_2$, if 
\[ \pi(\tilde\theta; \yobs) \approx \text{Est}_1 \gg \text{Est}_2, \]
then $\text{Est}_1$ must be better than $\text{Est}_2$.  

Table~\ref{tbl:sim} displays the average of estimates of $\uPi(A; \yobs)$ over $200$ simulated response vectors under Settings~1 and 2 with $3$ and $4$ predictor variables, respectively.  The main takeaway is that grid search loses approximation accuracy as the dimension increases while the stochastic optimization approaches do not appear to lose accuracy.  Grid search declines in efficiency compared to the constrained MLE when increasing from $3$ to $4$ predictors due to taking a coarser grid and using fewer Monte Carlo samples to approximate $\pi(\theta; \yobs)$ at each grid point.  We see little difference in performance between different variations of Kiefer--Wolfowitz updating. The approaches using averaged iterates produce better results than without averaging, which may suggest substantial variability in iterates remained after $500$ iterations in the second simulation under the second run of Kiefer-Wolfowitz updating.  

Stochastic optimization's performance was as good or better than grid search,  
arguably even better than the results indicate.  First, in these simulations, the grid for grid search was chosen favorably as a subset that contains the constrained MLE with high probability. In practice, the choice of grid is challenging, and may not be guaranteed to contain the maximizer.  Second, in order to compare the algorithms, we used no stopping criteria for stochastic optimization.  As the example in Section~\ref{ss:ex1} illustrates, the number of iterations needed for practical convergence of stochastic optimization can vary considerably.  So, the  number of Monte Carlo samples needed to produce the stochastic optimization results in Table~1 is generally less than the number of samples actually used.

\begin{table}[t]
\begin{center}
\begin{tabular}{lcc}
\multicolumn{1}{c}{}       & \multicolumn{2}{c}{Estimate of $\uPi(A; \yobs)$}        \\
\multicolumn{1}{c}{Method} & \multicolumn{1}{c}{Sim 1} & \multicolumn{1}{c}{Sim 2} \\ \hline
Constrained MLE            &      0.55 (0.35)                      &     0.52 (0.35)                       \\
Grid Search               &       0.49 (0.34)                     &    0.41 (0.33)                        \\
K--W $M = 16$, averaging      &    0.48 (0.38)                       &      0.51 (0.36)                      \\
no averaging               &     0.48 (0.38)                       &   0.49 (0.34)                         \\
K--W $M = 40$, averaging      &      0.50 (0.37)                      &     0.51 (0.36)                       \\
no averaging               &      0.50 (0.37)                      &     0.22 (0.14)                      
\end{tabular}
\end{center}
\caption{Average estimates (and standard deviations) of $\uPi(A; \yobs)$ for two variations of Kiefer--Wolfowitz updating, both with and without averaging iterates.  A grid search approximation of $\uPi(A; \yobs)$ and a Monte Carlo approximation of $\pi(\tilde\theta;\yobs)$ at the constrained MLE are displayed for comparison.}
\label{tbl:sim}
\end{table}

\section{Conclusion}
\label{S:discuss}

Several recent works have studied Monte Carlo approximation of upper and lower expectations by grid search methods.  In theory these methods are guaranteed to provide accurate approximations, but in practice accuracy is limited by specification of a subset of the parameter space in which to conduct grid search. It is possible the analyst may specify a subset not containing the true optimal $\theta^\star$, in which case grid search is inconsistent, regardless of the fineness of the grid.  However, even if the analyst makes a good choice of subset, accuracy is limited by the fineness of the grid for a fixed computational cost.  This ``curse of dimensionality'' causes practical accuracy to decline sharply as the number of parameters increases.  

Stochastic optimization offers provably accurate approximation of upper and lower expectations along with better practical performance than grid search in multi-dimensional problems.  It can be challenging to decide which variation of stochastic optimization to use and how to choose values of tuning parameters, but these choices are no more challenging for the analyst than subset selection for grid search.  In our examples, the Robbins--Monro and Kiefer--Wolfowitz procedures with iterate averaging and relatively large step-sizes worked best.  

In light of recent observations in \citet{martin.2021}, much of what is considered ``frequentist statistical inference'' can be formulated using notions of imprecise probability, in particular, possibility theory.  This means that problems like the one presented in Section~\ref{ss:logistic}, where lower and upper probabilities are evaluated based on optimization, are of fundamental importance for statisticians and data scientists.  Therefore, it is important to be able to solve these problems as accurately and efficiently as possible.  The stochastic optimization tools presented here seem quite promising, but more work is needed to develop (a) general rules for tuning the algorithms' parameters and (b) software that is easy to use.


\begin{thebibliography}{}

\bibitem[Agarwal et. al, 2012]{agarwal.etal}
Agarwal, A., Bartlett, P. L., Ravikumar, P., and Wainwright, M. J. (2012).
\newblock Information-Theoretic
Lower Bounds on the Oracle Complexity of Stochastic Convex Optimization.
\newblock {\em IEEE Trans. Inf. Theory.} 58(5):3235--3249.

\bibitem[Blum, 1954]{blum.1954}
Blum, J. R. (1954).
\newblock Approximation Methods which Converge with Probability one.
\newblock {\em Ann. Math. Statist.} 25(2):382--386.

\bibitem[Bouleau and Lepingle, 1994]{bouleau.1994}
Bouleau, N. and Lepingle, D. (1994).
\newblock {\em Numerical Methods for Stochastic Processes.}
\newblock New York: John Wiley.  

\bibitem[Byrd et. al, 2016]{byrd.etal}
Byrd, R. H., Hansen, S. L., Nocedal, J., and Singer, Y. (2016).
\newblock A Stochastic Quasi-Newton Method for Large-Scale Optimization.
\newblock {\em SIAM J. Optim.} 26(2):1008--1031.

\bibitem[Decadt, et al. 2019]{decadt.2019}
Decadt, A., de Cooman, G., and De Bock, J. (2019).
\newblock Monte Carlo Estimation for Imprecise Probabilities: Basic Properties.
\newblock Proceedings of the Eleventh International Symposium on Imprecise Probabilities: Theories and Applications, in {\em Proc. Mach. Learn. Res.} 103:135--144. 
 
\bibitem[Fetz and Oberguggenberger, 2016]{fetz.2016}
Fetz, T. and Oberguggenberger, M. (2016).
\newblock Imprecise random variables, random sets, and Monte Carlo simulation.
\newblock {\em Int. J. Approx. Reason.} 78:252--264. 
 
\bibitem[Fetz, 2019]{fetz.2019}
Fetz, T. (2019).
\newblock Improving Convergence of Iterative Importance Sampling for Computing Upper and Lower Expectations.
\newblock Proceedings of the Eleventh International Symposium on Imprecise Probabilities: Theories and Applications, in {\em Proc. Mach. Learn. Res.} 103:185--193. 

\bibitem[Fletcher, 1987]{fletcher1987}
Fletcher, R. (1987). 
\newblock {\em Practical Methods of Optimization}. 2nd ed. 
\newblock New York: John Wiley \& Sons. 

\bibitem[Kiefer and Wolfowitz, 1952]{kiefer.1952}
Kiefer, J. and Wolfowitz, J. (1952).
\newblock Stochastic Estimation of the Maximum of a Regression Function.
\newblock {\em Ann. Math. Statist.} 23(3):462--466.

\bibitem[Kushner and Yin, 2003]{kushner.2003}
Kushner, H. J., and Yin, G. G. (2003).
\newblock {\em Stochastic Approximation
and Recursive Algorithms
and Applications.} 2nd ed.
\newblock New York: Springer.  

\bibitem[Liu and Martin, 2020]{martin.liu.poss}
Liu, C. and Martin, R. (2020).
\newblock Inferential models and possibility measures.
\newblock \url{https://researchers.one/articles/20.08.00004}.

\bibitem[Martin, 2015]{martin.2015}
Martin, R. (2015).
\newblock Plausibility functions and exact frequentist inference.
\newblock {\em J. Amer. Stat. Assoc.} 110:1552--1561.

\bibitem[Martin, 2018]{martin.2018}
Martin, R. (2018).
\newblock On an inferential model construction using generalized associations.
\newblock {\em J. Stat. Plan. Infer.} 195:105--115.

\bibitem[Martin, 2021]{martin.2021}
Martin, R. (2021).
\newblock An imprecise-probabilistic characterization of frequentist statistical inference
\newblock \url{https://researchers.one/articles/21.01.00002}.

\bibitem[Martin and Ghosh, 2008]{martin.2008}
Martin, R. and Ghosh, J. K. (2018).
\newblock Stochastic Approximation and Newton’s Estimate of a Mixing Distribution.
\newblock {\em Stat. Sci.} 23(3):365--382.

\bibitem[Martin and Liu, 2015]{martin.liu.book}
Martin, R. and Liu, C. (2015).
\newblock {\em Inferential Models: Reasoning with Uncertainty.}.
\newblock Monographs in Statistics and Applied Probability Series, Chapman \& Hall/CRC Press.

\bibitem[Polyak and Juditsky, 1992]{polyak}
Polyak, B. T and Juditsky, A. B. (1992).
\newblock Acceleration of stochastic approximation by averaging.
\newblock {\em SIAM J. Control Optim.} 30:838--855.

\bibitem[Robbins and Monro, 1951]{robbins.1951}
Robbins, H. and Monro, S. (1951).
\newblock A Stochastic Approximation Method.
\newblock {\em Ann. Math. Statist.}. 22(3):400--407. 

\bibitem[Robbins and Siegmund, 1971]{robbins.1971}
Robbins, H. and Siegmund, D. (1971).
\newblock A Convergence Theorem for Non-negative Almost Supermartingales and Some Applications.
\newblock In {\em Optimizing Methods in Statistics}, ed. Rustagi, J. S. Academic Press, 233--257. 
\bibitem[Sacks, 1958]{sacks.1958}
Sacks, J. (1958).
\newblock Asymptotic Distribution of Stochastic Approximation Procedures.
\newblock {\em Ann. Math. Statist.} 29(2):373--405.

\bibitem[Zeiler, 2012]{zeiler}
Zeiler, M. D. (2012).
\newblock ADADELTA: An Adaptive Learning Rate Method.
\newblock \url{https://arxiv.org/abs/1212.5701}.


\end{thebibliography}
\end{document}